\documentclass{article}

\usepackage{amsfonts}
\usepackage{amsmath}
\usepackage{amssymb}
\usepackage{amsthm}

\newtheorem{thm}{Theorem}[section]

\newtheorem{lem}[thm]{Lemma}

\theoremstyle{definition}
\newtheorem{defn}[thm]{Definition}
\newtheorem{rem}[thm]{Remark}

\begin{document}

\title{An Application of Group Theory in Confidential Network Communications \thanks{The Research was supported in part by the Swiss National Science
Foundation under grant No.\@ 169510. First author is partially
suppported by  Ministerio de Economia y Competitividad grant
MTM2014-54439 and Junta de Andalucia (FQM0211). The last author is
supported by Armasuisse.
} } %$\thanks{\jobname.tex}$}

\author{
J.A.~L\'opez-Ramos\footnote{University of Almeria},
J.~Rosenthal\footnote{University of Zurich},
D.~Schipani\footnotemark[3], R.~Schnyder\footnotemark[3] }

\maketitle

\begin{abstract}
A new proposal for group key exchange is introduced which proves to be both efficient and secure and compares favorably with state of the art protocols.
\end{abstract}

\section{Introduction}

Group Key Exchange (GKE) is currently a highly relevant topic
 due to the explosion of group communications  in many applications that provide information exchange (cf.
\cite{lee} and \cite{vandenmerwe} and their references). A very
recent application can be found in wireless sensor networks (WSNs),
core of the so-called Internet of Things, that consists of tiny
autonomous low-cost low-power devices that carry out monitoring
tasks. WSNs can be found in many civil applications. The devices are
called sensor nodes and the monitored data is typically collected at
a base station, that will be later processed in data mining servers.

\medskip

Since Ingemarsson et al.\@ in \cite{ingemarsson} made an attempt to
extend Diffie-Hellman two-party key exchange given in \cite{diffie},
many works have dealt with this issue, i.e., providing protocols for a
group of communication nodes that allow this group to build a
common key in a distributed and collaborative manner. In
\cite{steer} the authors gave a distributed protocol for GKE that
does not run efficiently in the Initial Key Agreement (IKA), which
is, in most cases, the main problem. Probably  two of the best
known distributed protocols are given in \cite{steiner1} and
\cite{steiner2}, and \cite{burmester1} and \cite{burmester2},
respectively, and both extend naturally the foundational
Diffie-Hellman key exchange.

\medskip

The protocol introduced in \cite{burmester1} and \cite{burmester2}
proves to be very efficient in the IKA, using just two rounds.
However, further rekeying operations require executing the protocol
completely as in the initial phase (IKA) and produce a big number
of messages, using a large bandwidth as the number of
communication nodes grows.

On the other hand, the protocol introduced in \cite{steiner1} and
\cite{steiner2} is quite efficient in the rekeying processes, i.e.,
rekeying operations once the group has shared a first key. The authors
give an Auxiliary Key Exchange (AKA) that makes use of a single
broadcasted message, although the IKA protocol is considerably
slower than the preceding one since it requires as many rounds as
the number of participants.

\medskip

But the main issue in both efficient proposals is their security. In
\cite{schnyder} and \cite{attackbur} active attacks to the systems
proposed in \cite{steiner1} and \cite{steiner2}, and
\cite{burmester1} and \cite{burmester2} respectively are introduced.
The authors show the possibility of sharing a common key with the
components of a communication group, without letting them notice
anything. However, one of the IKA proposals introduced in
\cite{steiner1} and \cite{steiner2} (the so-called IKA.1) avoids
this attack, but, as pointed out above, it requires a big numbers of
rounds as the number of users grows.

\medskip

In this paper we are introducing a new proposal that avoids both attacks (\cite{attackbur} and
\cite{schnyder}) and shows the best characteristics of the
aforementioned proposals: on one hand, the key is obtained in a
distributed IKA with just two rounds and, on the other, the AKA
protocol that provides rekeying operations is developed by means of a
single message. Our protocol extends naturally the Diffie-Hellman
key exchange as well and we show that its security is based on the
difficulty of problems that refer to decisions in
the group where the two party Diffie-Hellman key exchange takes
place. %: the Computational Diffie-Hellman problem and the Decisional
%Diffie-Hellman problem.
%The Computational Diffie-Hellman (CDH) problem in a group $G$ states
%that given $g$ in $G$ and $g^x$ and $g^y$ for some integers $x$ and
%$y$, find $g^{xy}$.
 Namely, the Decisional Diffie-Hellman (DDH) problem in a group $G$ is the problem to decide whether,
given $g$ in $G$ and random $x,y,z$ in $G$,
$z$ equals $x^{\log_g y}$.

%\begin{problem} (\cite{maze})
%Given a finite abelian semigroup $G$ acting on a finite set $S$ and
%elements $x, y, z \in S$ with $y = g \cdot x$ and $z =  h \cdot x$
%for some $g, h \in G$, find $(gh)\cdot x$.
%\end{problem}

%Throughout this paper $G$ and $H$ will denote abelian semigroups. We
%will say that an action $\Phi: G\times H \rightarrow H$ defined by
%$\Phi (g,s)=g\cdot s$ is linear in case $\Phi(g,ss')=g\cdot
%ss'=(g\cdot s)(g\cdot s')=\Phi(g,s)\Phi (g,s')$. Here we denote by
%yuxtaposition operations in both semigroups.

\section{The Initial Key Agreement}

Let us start by establishing the general setting for the GKE
protocol. Participants in the communication process will be given by
the set $\{ U_1, \ldots  U_n \}$. The group of users agree on a
cyclic group $G$ of order a prime $q$ and a generator $g$ of $G$.

\medskip

Every participant ${\cal U}_i$ holds two pairs of private-public
keys, say $(r_i,g^{r_i})$ and $(x_i,g^{x_i})$. One of these users
will be the group controller that we will denote by ${\cal
U}_{c_1}$, for some $c_1$ in the set $\{ 1, \ldots ,n\}$. He will be
in charge of sending the initial keying information as well as the
following rekeying messages in case we wish to define a centralized
protocol. However, as we will see in the following section, the
character of the protocol can change from centralized to distributed
(and vice versa) at any point of the following rekeying stages. The
protocol that describes the initial key agreement is given by the
following steps.

\medskip

\noindent {\it Protocol 1:}

\medskip

{\bf First Round:}

\begin{enumerate}

\item Every user ${\cal U}_i$ publishes his pair of public keys $(g^{r_i}, g^{x_i})$, $i=1,
\ldots ,n$, $i\not= c_1$.

\item The group controller ${\cal U}_{c_1}$ computes the key
$K_1=g^{r_{c_1}\sum_{j=1, j\not= c_1}^nr_j}$.

\item The group controller takes a new pair of private elements $(r'_{c_1}, x'_{c_1})$,
which becomes his new private information. This will be used in the case of a rekeying operation at a later stage.

\noindent {\bf Second Round:}

\item Every user, using the public
information, ${\cal U}_i$, $i=1, \ldots ,n$, $i\not= c_1$, computes
$g^{\sum_{j=1, j\not= c_1,i}^n r_j}$ and sends this value to $U_{c_1}$.

\item The group controller ${\cal U}_{c_1}$ broadcasts the keying message
$$\{ Y_{1,1}, \ldots , Y_{1,c_1}, \ldots ,Y_{1,n}, R_1, S_1 \}$$

\noindent where $Y_{1,i}=g^{-x_{c_1}x_i}\left (
g^{r_{c_1}\sum_{j\not= c_1,i} r_j }\right )$, for $i=1, \ldots ,n$,
$i\not= c_1$,

\noindent $Y_{1,c_1}=K_1g^{-r'_{c_1}r_{c_1}}g^{-x'_{c_1}x_{c_1}}$,
and $R_1=g^{r_{c_1}}$ and $S_1=g^{x_{c_1}}$.

\item Every user ${\cal U}_i$ computes $K_{1,i}=Y_{1,i} S_1^{x_i}R_1^{r_i}$, $i=1, \ldots ,n$, $i\not= c_1$.
\end{enumerate}

The proof of the following Lemma is straightforward and shows the
correctness of the protocol.

\begin{lem}
$K_{1,i}=K_1$ for every $i=1, \ldots ,n$, $i\not= c_1$.
\end{lem}

\begin{rem}
Let us assume that the number of users is $n=2$ and that $g^{x_i}=e$
for $i=1,2$, where $e\in G$ is the neutral element. Now if ${\cal U}_1$ makes public $g^{r_1}$ in the first
round, ${\cal U}_2$ will send the keying message $\{ e, R_1=g^{r_2}
\}$ in the second round. Thus our protocol is a natural extension
of  the classical Diffie-Hellman key exchange in the group $G$.
\end{rem}

It can be observed in the preceding protocol that user ${\cal
U}_{c_1}$ bears most of the workload. The protocol is
designed in such a way that every node publishes just  a pair of public
keys, while ${\cal U}_{c_1}$ computes what is required for the first
keying. This could be the case when ${\cal U}_{c_1}$ is a server
that processes the pieces of information transmitted by every user. However,
in case every user has similar capabilities, we can slightly modify
the preceding protocol and distribute the computation requirements. As
previously, every user holds a pair of private keys $(r_i,x_i)$.

\medskip

\noindent {\it Protocol 2:}

\medskip

{\bf First Round:}

\begin{enumerate}

\item Every user ${\cal U}_i$ publishes his public key $g^{r_i}$, $i=1,
\ldots ,n$, $i\not= c_1$.

\item The group controller ${\cal U}_{c_1}$ computes the key
$K_1=g^{r_{c_1}\sum_{j=1, j\not= c_1}^nr_j}$.

\item The group controller takes a new pair of private elements $(r'_{c_1}, x'_{c_1})$,
which becomes his new private information.

\noindent {\bf Second Round:}

\item Every user, using the public information, ${\cal U}_i$, $i=1,
    \ldots ,n$, $i\not= c_1$, computes $g^{\sum_{j=1,
j\not= c_1,i}^n r_j}g^{-x_i}$ and sends this value to $U_{c_1}$.

\item The group controller ${\cal U}_{c_1}$ broadcasts the keying message
$$\{ Y_{1,1}, \ldots , Y_{1,c_1}, \ldots ,Y_{1,n}, R_1 \}$$

\noindent where $Y_{1,i}=\left (
g^{r_{c_1}\sum_{j=1, j\not= c_1,i}r_j}\right )g^{-r_{c_1}x_i}$, for
$i=1, \ldots ,n$, $i\not= c_1$,

\noindent $Y_{1,c_1}=K_1g^{-r'_{c_1}r_{c_1}}g^{-x'_{c_1}r_{c_1}}$,
and $R_1=g^{r_{c_1}}$.

\item Every user ${\cal U}_i$ computes $K_{1,i}=Y_{1,i} R_1^{x_i}R_1^{r_i}$, $i=1, \ldots ,n$, $i\not= c_1$.
\end{enumerate}

We will now state the security of the preceding protocols. To this
end let us now recall the following definition.

\begin{defn}\cite[Definition 2.2]{burmester2}
Let ${\cal P}$ be a GKE protocol and ${\cal A}$ a passive adversary.
Assume that ${\cal A}$ has witnessed polynomially-many instances of
${\cal P}$ and let $K$ be the key output by the last instance.

%We will say that ${\cal P}$ guarantees privacy if it is
%computationally infeasible for ${\cal A}$ to compute $K$.

We will say that ${\cal P}$ guarantees secrecy if ${\cal A}$ cannot
distinguish $K$ from a random bit string of the same length with
probability better than $1/2 +\varepsilon$, where $\varepsilon$ is
negligible
\end{defn}

\begin{thm}\label{secrecy}
If the DDH  problem is intractable, then Protocols 1 and 2 provide
 secrecy.
\end{thm}
\begin{proof}
We observe that we can see the broadcasted message in
Protocol 1 as a multiple ElGamal type of encryption in the following
way.  For $i\neq c_1$ we first encrypt $K_1$ using the public value
$g^{r_i}$ and $r_{c_1}$ as random parameter,
obtaining
$\big( X_{1,i}, R_1 \big) =
\big( g^{r_{c_1}\sum_{j\not= c_1,i} r_j }, g^{r_{c_1}} \big)$ and
then we encrypt $X_i$ using the public key $g^{x_{i}}$ and
$x_{c_1}$ as a random parameter, obtaining the pair $\big( Y_{1,i},
S_1 \big)$.

The case of $Y_{1,c_1}$ is analogous using the elements
$g^{r'_{c_1}}$ and $g^{x'_{c_1}}$, that, although unknown
to a passive adversary, could also be made public.

Now using Lemma 1 and Theorem 1 of \cite{bellare} we
can deduce the thesis.

The security of Protocol 2 follows similarly.

\end{proof}

% {\bf I AM NOT SURE ABOUT ELEMENT $eg^bg^{x_{c_1}}$ MY IDEA IS TO USE
% THE SAME REASONING AS THAT APPEARING ON PAGE 4 OF THE PREVIOUS
% VERSION WHERE IT IS SHOWN THE EQUIVALENCE OF ELGAMAL AND DHSAP}

%Let us remark at this point that Protocol 1 requires a group
%controller to build the first key. This user ${\cal U}_{c_1}$ could
%be in charge of the next rekeying messages as we will show in the
%next section and thus he will be the first obtaining the session
%key. In this case we are dealing with a centralized scheme and thus
%$Y_{1,c_1}$ is not needed to be sent since this will be the user
%${\cal U}_{c_1}$'s rekeying information as we will show later.

\begin{rem}
For ease of notation we have presented the preceding protocols using the action $\Phi
(y,g^x)=(g^x)^y$, but more general scenarios based on linear actions can be considered too, as in \cite{groupkey}.
\end{rem}

\begin{rem}
We could wonder about the necessity of using two different keys for
every user. To clarify this  suppose only one key is used. %let us consider the following example:
%Let $S$ be a cyclic group of order $q$ generated by $g$, with the
%action $\Phi : \mathbb{N}^* \times S \rightarrow S$ given by $\Phi
%(y,g^x)=(g^x)^y$.
%Then Protocol 1 implies sharing
Then we would share a key of the form
$K = g^{k_j\sum_{r=1,r\not=j}^nk_r}$. Without the $x_i$, an adversary could access the
messages
$$D_i=g^{k_j\sum_{r=1,r\not=i,j}^nk_r},  \ i\not= j, \ i=1, \dots ,n,$$
from which she can compute $\prod_{r=1,r \ne j}^{n} D_r=K^{n-2}$. In
the case where the order $q$ of $S$ is known, the adversary can now
recover the key $K$ from $K^{n-2}$ by inverting $n-2$ modulo $q$.
This is in particular the case where $G$ is a subgroup of a finite
field, or where it is the group of points of an elliptic curve.
Using two different keys for every user avoids the above situation
in these cases. However, the use of a single key for every user could still apply in other settings \cite{groupkey}.
% the following setting: Let $m=pq$ with $p$ and $q$ two large primes and
%let $G = \mathbb{Z}_{(p-1)(q-1)}^*$. Then the action $\Phi :G\times
%\mathbb{Z}_m \rightarrow \mathbb{Z}_m$ given by $\Phi (x,g)=g^x \
%\mbox{mod} \ m$ shows an example where the above attack cannot be
%developed unless the adversary is able to factorize $m$. The shared
%key in this case is of the form $g^{x_j\sum_{i=1,i\not=j}^nx_i} \
%\mbox{mod} \ m$.
\end{rem}

% The fact that the protocol, as noted above needs to use two
% different private keys for every user gives also the opportunity to
% show another security property of both protocols.

% Firstly we observe that we can see the broadcasted message in
% Protocol 1 as a multiple encryption of the messages
% $K_1g^{-r_{c_1}r_i}$ using the public key $g^{x_i}$ for every $i=1,
% \ldots ,n$, $i\not= c_1$ and the common random parameter $x_{c_1}$,
% and a encryption of the message $K_1g^{-r'_{c_1}r_{c_1}}$ using the
% public key (although it is not made public) $g^{x'_{c_1}}$ and
% $x_{c_1}$ as a random parameter (analogously for that broadcast of
% Protocol 2). Now using Lemma 1 and Theorem 1 of \cite{bellare} we
% can deduce the following result.

% \begin{thm}\label{secrecy}
% If the DDH problem is intractable, then Protocols 1 and 2 provide
% secrecy.
% \end{thm}

\section{The Auxiliary Key Agreement}

In the preceding section we have introduced a protocol to build a
common key based on the information held by every user. As we can
observe this may require some computational resources on one of the
participants. On the other hand, this session key may expire due
simply to key caducity or to changes in the communication group, i.e.,
users may join or leave the group and we are concerned about preserving
secrecy of previous, respectively, later communications. In this
section we provide an Auxiliary Key Agreement that solves this matter
very efficiently and, moreover, allows either to keep the centralized
aspect of the GKE, or to change to a distributed scheme, allowing
any user to provide a new key for the remaining members of the group.

\medskip

In the more general setting, we are assuming that some rekeying
rounds have occurred,
and in the next step keys are to be renewed again possibly by a new controller.
%made by the same user ${\cal U}_{c_1}$ in case
%we have dealt with a centralized scheme and thus the rekeying
%information for user ${\cal U}_{c_1}$ was not needed during those
%rounds, either in the distributed case, and thus the rekeying
%message is sent by any of the remainder users, so after some round,
%the rekeying information for user ${\cal U}_{c_1}$ is needed.
The
following protocol shows the Auxiliary Key Agreement %, as above cited,
after $t-1$ rekeying rounds, $t>1$, whereby $K_t$ denotes the last common
key shared by the group. The user in charge of the $t$-th rekeying
will be user ${\cal U}_{c_t}$, distinct from the preceding
controller, and thus, rekeying information of this will be needed.
Without loss of generality, we may assume that the precedent
controller was user ${\cal U}_{c_1}$, the key shared was $K_{t-1}$ and the last rekeying
message was
$$\{ Y_{t-1,1}, \ldots ,Y_{t-1,c_1}, \ldots ,Y_{t-1,n}, R_{t-1}, S_{t-1} \}$$

\noindent where $Y_{t-1,c_1}=K_{t-1}g^{-r'_{c_1}r_{c_1}}g^{-x'_{c_1}x_{c_1}}$, following Protocol 1.
The proposed AKA is then as follows:
\noindent {\it Protocol 3:}

\begin{enumerate}

\item User ${\cal U}_{c_t}$ computes two new elements $r'_{c_t}$ and
$x'_{c_t} \in G$ that become his new private information.

\item User ${\cal U}_{c_t}$ computes the new session key
$ K_{t-1}^{r'_{c_t}}$.

\item User ${\cal U}_{c_t}$ broadcasts the rekeying message
$$\{ Y_{t,1}, \ldots , Y_{t,c_t}, \ldots ,Y_{t,n}, R_t, S_t \}$$

\noindent where $Y_{t,i}=Y_{t-1,i}^{r'_{c_t}}$, $i\neq c_t$,

\noindent $Y_{t,c_t}=K_t R_{t-1}^{-r'_{c_t}r'_{c_t}} S_{t-1}^{-r'_{c_t}x'_{c_t}}$,

\noindent and $R_t=R_{t-1}^{r'_{c_t}}$ and $S_t=
S_{t-1}^{r'_{c_t}}$.

\item Every user ${\cal U}_i$ computes $K_{t,i}=Y_{t,i}S_t^{x_i}
R_t^{r_i}$, $i=1, \ldots ,n$, $i\not= c_t$.

\end{enumerate}

\begin{lem}\label{correct}
$K_{t,i}=K_t$ for every $i=1, \ldots ,n$.
\end{lem}

\begin{proof}
%Let us assume as induction hypothesis that $Y_{t-1,i}$ allows user
%${\cal U}_i$ to recover $K_{t-1}$. Then,
For $i=1, \ldots ,n$,
$i\not= c_t$

$$\begin{array}{rl} K_{t,i} & = Y_{t,i}S_t^{x_i}R_t^{r_i}
\\
& =  Y_{t-1,i}^{r'_{c_t}} S_{t-1}^{x_i r'_{c_t}} R_{t-1}^{r_i r'_{c_t}} \\
& =  (Y_{t-1,i} S_{t-1}^{x_i } R_{t-1}^{r_i })^{r'_{c_t}}   \\
& =  K_{t-1}^{r'_{c_t}}\\
& =  K_t.
\end{array}$$

%\noindent and analogously for $K_{t,c_1}$.
\end{proof}

%Now, using the same argument given in Theorem \ref{privacy} we get
%the following.

%\begin{thm}
%Protocol 2 guarantees privacy if and only if the DHSAP problem is
%intractable.
%\end{thm}

\subsection{Altering the membership}

As it was previously pointed out members of the group can be
constantly changing: some of them may leave the group and other may
wish to join it. For this reason it is convenient that those joining
the group should not be able to get previously distributed keys and
those leaving it should not get future distributed keys in order to
preserve confidentiality of former and future communications.

\medskip

In the case of a set of members leaving the group, rekeying is made
naturally following Protocol 2, but erasing those positions in the
rekeying message corresponding to those users leaving the group.

\medskip

On the other hand joining operation can be carried out in a massive
way and by any of the users. Let us assume that after the $t$-th
rekeying message, $l$ users wish to join the group. Then the
corresponding values $R_t$ and $S_t$ are made public and without
loss of generality, we may suppose that these $l$ users sent a
petition to join the group to a user that will be ${\cal
U}_{c_{t+1}}$ for some $c_{t+1}$ in the set $\{ 1, \ldots ,n\}$.

In case the petitions are sent to different users, each of them may
collect the received petitions and send a rekeying message including
the petitioners after a period of time. The protocol is as follows:

\noindent {\it Protocol 4.}

\begin{enumerate}

\item Every new user ${\cal U}_{n+j}$, $j=1, \ldots ,l$,  sends a petition to user
${\cal U}_{c_{t+1}}$ jointly with the pair $R_t^{r_{n+j}},
S_t^{x_{n+j}}$, where $r_{n+j}, x_{n+j}\in G$ is the user ${\cal
U}_{n+j}$'s private information.

\item User ${\cal U}_{c_{t+1}}$ computes two new elements $r'_{c_{t+1}},
x'_{c_{t+1}}\in G$ that become his new private information.

\item User ${\cal U}_{c_{t+1}}$ computes the new key
$K_{t+1}= \big(
K_t \displaystyle{ R_t^{\sum_{i=1}^l r_{n+i}}}\big) ^{r'_{c_{t+1}}}$.

\item User ${\cal U}_{c_{t+1}}$ broadcasts the rekeying message

$$\{ Y_{t+1,1}, \ldots , Y_{t+1,c_{t+1}}, \ldots , Y_{t+1,n},
Y_{t+1,n+1}, \ldots , Y_{t+1,n+l}, R_{t+1}, S_{t+1} \}$$

\noindent where $Y_{t+1,i}=
\big(
Y_{t,i}\displaystyle{ R_t^{\sum_{i=1}^l r_{n+i}}}\big) ^{r'_{c_{t+1}}}$,
for $i=1, \ldots ,n$, $i\not= c_{t+1}$,

\noindent $Y_{t+1,c_{t+1}}=K_{t+1} R_t^{-r'_{c_{t+1}}r'_{c_{t+1}}}
S_t^{-r'_{c_{t+1}}x'_{c_{t+1}}}$,

\noindent $Y_{t+1,i}= K_{t+1} R_t^{-r_{i}r'_{c_{t+1}}}
S_t^{-x_{i}r'_{c_{t+1}}}$, for $i=n+1,
\ldots ,n+l$,

\noindent $R_{t+1}=R_t^{r'_{c_{t+1}}}$ and
$S_{t+1}=S_t^{r'_{c_{t+1}}}$.

\item Every user ${\cal U}_i$ computes $K_{t+1,i}=Y_{t+1,i}S_{t+1}^{x_i}
R_{t+1}^{r_i}$, $i=1, \ldots ,n+l$, $i\not= c_{t+1}$.

\end{enumerate}

Analogously to Lemma \ref{correct} we get the following.

\begin{lem}
After Protocol 3, $K_{t+1,i}=K_{t+1}$ for every $i=1, \ldots ,n+l$,
$i\not= c_{t+1}$.
\end{lem}

Also, using the same argument as in the IKA protocol, we may show that protocols
for rekeying and altering the membership provide secrecy for the former and
future keys distributed.

\end{document}